\documentclass[11pt]{article}

\title{Optimal Unateness Testers for Real-Valued Functions:\\ Adaptivity Helps\footnote{Preliminary (much weaker) versions of this paper were posted as~\cite{CS16unate} and~\cite{BMPR16}.}}

\author{Roksana Baleshzar
\footnote{Computer Science and Engineering, Pennsylvania State University. Email: {\sf rxb5410@cse.psu.edu, ramesh@psu.edu, sofya@cse.psu.edu}. The work of these authors was partially supported by NSF award CCF-1422975. }
\and 
Deeparnab Chakrabarty
\footnote{Microsoft Research, Bangalore. Email: {\sf deeparnab@gmail.com}.} 
\and 
Ramesh Krishnan S. Pallavoor
\footnotemark[2] 
\and 
Sofya Raskhodnikova
\footnotemark[2] 
\and 
C. Seshadhri
\footnote{Computer Science, University of California, Santa Cruz. Email: {\sf sesh@ucsc.edu.}}
}

\date{}

\usepackage{paralist}
\usepackage{amsmath,amssymb,amsthm,mathtools}
\usepackage{xspace}
\usepackage{color}
\usepackage{xcolor}
\usepackage{fullpage}

\usepackage[colorlinks, citecolor = black,
urlcolor=black,linkcolor=black,pdfstartview=FitH]{hyperref}
\usepackage[ruled, noend, noline]{algorithm2e}
\usepackage{verbatim}
\usepackage{bbm}

\usepackage{chngcntr}
\usepackage{apptools}

\newtheorem{theorem}{Theorem}[section]
\newtheorem{definition}[theorem]{Definition}

\newtheorem{lemma}[theorem]{Lemma}
\newtheorem{claim}[theorem]{Claim}

\AtAppendix{\counterwithin{theorem}{section}}

\newcommand{\eps}{\varepsilon}

\def\R{{\mathbb R}}
\def\bb{\mathbf{b}}

\newcommand{\ord}[2][th]{\ensuremath{{#2}^{\mathrm{#1}}}}
\newcommand{\coord}[2]{t^{#1}_{#2}}
\newcommand{\capt}{\mathtt{cap}}
\def\Yes{\mathbf{Yes}}
\def\No{\mathbf {No}}
\def\dec{\mathsf{val}}
\def\sgn{\mathsf{sgn}}
\def\supp{\mathsf{supp}}
\def\dd{{d'}}
\def\totcube{m}
\def\dir{\mathsf{dir}}

\newcommand{\dist}{\mathsf{dist}}

\newcommand{\Exp}{\EX}

\newcommand{\cei}[1]{\lceil#1\rceil}
\newcommand{\EX}{\hbox{\bf E}}

\newcommand{\cA}{{\cal A}}

\newcommand{\cC}{{\cal C}}

\newcommand{\cG}{\mathcal{G}}

\newcommand{\cP}{\mathcal{P}}

\newcommand{\cT}{{\cal T}}

\newcommand{\N}{\mathbb N}
\newcommand{\calG}{{\cal G}}
\newcommand{\calA}{{\cal A}}
\newcommand{\calC}{{\cal C}}
\newcommand{\calE}{{\cal E}}

\newcommand{\bG}{\boldsymbol{G}}

\newcommand{\bS}{\boldsymbol{S}}
\newcommand{\bT}{\boldsymbol{T}}

\newcommand{\Sec}[1]{\hyperref[sec:#1]{Section~\ref*{sec:#1}}} 
\newcommand{\Eqn}[1]{\hyperref[eq:#1]{(\ref*{eq:#1})}} 
\newcommand{\Fig}[1]{\hyperref[fig:#1]{Fig.\,\ref*{fig:#1}}} 
\newcommand{\Tab}[1]{\hyperref[tab:#1]{Tab.\,\ref*{tab:#1}}} 
\newcommand{\Thm}[1]{\hyperref[thm:#1]{Theorem~\ref*{thm:#1}}} 
\newcommand{\Fact}[1]{\hyperref[fact:#1]{Fact\,\ref*{fact:#1}}} 
\newcommand{\Lem}[1]{\hyperref[lem:#1]{Lemma~\ref*{lem:#1}}} 
\newcommand{\Prop}[1]{\hyperref[prop:#1]{Prop.~\ref*{prop:#1}}} 
\newcommand{\Cor}[1]{\hyperref[cor:#1]{Corollary~\ref*{cor:#1}}} 
\newcommand{\Conj}[1]{\hyperref[conj:#1]{Conjecture~\ref*{conj:#1}}} 
\newcommand{\Def}[1]{\hyperref[def:#1]{Definition~\ref*{def:#1}}} 
\newcommand{\Alg}[1]{\hyperref[alg:#1]{Algorithm~\ref*{alg:#1}}} 
\newcommand{\Ex}[1]{\hyperref[ex:#1]{Ex.~\ref*{ex:#1}}} 
\newcommand{\Clm}[1]{\hyperref[clm:#1]{Claim~\ref*{clm:#1}}} 
\newcommand{\Step}[1]{\hyperref[step:#1]{Step~\ref*{step:#1}}} 

\begin{document}

\maketitle
\begin{abstract}
We study the problem of testing unateness of functions $f:\{0,1\}^d\to\R.$
We give a $O(\frac d \eps \cdot \log\frac d\eps)$-query nonadaptive tester and a
$O(\frac d\eps)$-query adaptive tester and show that both testers are optimal for a fixed distance parameter $\eps$. Previously known unateness testers worked only for Boolean functions, and their query complexity had worse dependence on the dimension both for the adaptive and the nonadaptive case. Moreover, no lower bounds for testing unateness were known\footnote{Concurrent work by Chen et al.~\cite{CWX17} proves an $\Omega(d/\log^2 d)$ lower bound on the nonadaptive query complexity of testing unateness of Boolean functions. Our stronger lower bounds are for real valued functions.}.
We also generalize our results to obtain optimal unateness testers for functions $f:[n]^d\to\R$.

Our results establish that adaptivity helps with testing unateness of real-valued functions on domains of the form $\{0,1\}^d$ and, more generally, $[n]^d$.
This  stands in contrast to the situation for monotonicity testing where there is no adaptivity gap for functions $f:[n]^d\to\R$.
\end{abstract}

\section{Introduction}
We study the problem of testing whether a given real-valued function $f$ on domain $[n]^d$, where $n,d\in\N,$ is unate.
A function  
$f:[n]^d \to \R$
is {\em unate} if for every coordinate $i\in [d]$, the function is either nonincreasing in the $\ord{i}$ coordinate or nondecreasing in the $\ord{i}$ coordinate.
Unate functions naturally generalize monotone functions, which are nondecreasing in all coordinates, and $\bb$-monotone functions, which have a particular direction in each coordinate (either nondecreasing or nondecreasing), specified by a bit-vector $\bb\in \{0,1\}^d$. More precisely,
 a function is $\bb$-monotone if it is nondecreasing in coordinates $i$ with $\bb_i = 0$ and nonincreasing in the other coordinates. Observe that a function $f$ is unate iff there exists some $\bb\in \{0,1\}^d$ for which $f$ is $\bb$-monotone.

A {\em tester}~\cite{RS96,GGR98} for a property $\cal P$ of a function $f$ is an algorithm that gets a distance parameter $\eps\in(0,1)$ and query access to $f$. It has to accept with probability at least $2/3$ if $f$ has property $\cal P$ and reject with probability at least $2/3$ if $f$ is $\eps$-far (in Hamming distance) from $\cal P$. We say that $f$ is $\eps$-far from $\cal P$ if at least an $\eps$ fraction of values of $f$ must be modified to make $f$ satisfy $\cal P$. A tester has {\em one-sided error} if it always accepts a function satisfying $\cal P$, and has {\em  two-sided
error} otherwise.
A {\em nonadaptive} tester makes all its queries at once, while an {\em adaptive} tester can make queries after seeing answers to the previous ones.

Testing of various properties of functions, including
monotonicity
(see, e.g., \cite{GGLRS00,DGLRRS99,EKKRV00,LR01,FLNRRS02,Fis04,HK07,AC06,HK08,BRW05,BBM12,BGJRW12,BCGM12,CS13, BlaRY14,BerRY14,CS14,CS16,CDJS15,CST14,CDST15,KMS15,BB15,BB16,DRTV16,PRV17} and recent surveys~\cite{Ras16,C16a}),
the Lipschitz property \cite{JR13,CS13,BlaRY14}, bounded-derivative properties~\cite{CDJS15}, and unateness~\cite{GGLRS00,KS16}, has been studied extensively over the past two decades. Even though unateness testing was initially discussed in the seminal paper
by Goldreich et al.~\cite{GGLRS00} that gave first testers for properties of functions,  relatively little is known about testing this property. All previous work on unateness testing focused on the special case of Boolean functions on domain $\{0,1\}^d$. The domain $\{0,1\}^d$ is called the {\em hypercube} and the more general domain $[n]^d$ is called the {\em hypergrid}.
Goldreich et al.~\cite{GGLRS00}
provided a $O(\frac{d^{3/2}}{\eps})$-query nonadaptive tester for unateness of Boolean functions on the hypercube.
Recently, Khot and Shinkar~\cite{KS16} improved the query complexity
to $O(\frac{d\log d}\eps)$, albeit with an adaptive tester.

In this paper, we improve upon both these works, and our results hold for a more general class of functions.
Specifically, we show that unateness of real-valued functions on hypercubes can be tested nonadaptively with $O(\frac d \eps \log \frac d \eps)$ queries and adaptively with $O(\frac d \eps)$ queries. More generally, we describe a $O(\frac d \eps \cdot(\log\frac d\eps + \log n))$-query nonadaptive tester and a
$O(\frac{d\log n}\eps)$-query adaptive tester of unateness of real-valued functions over hypergrids.

In contrast to the state of knowledge for unateness testing, the complexity of testing monotonicity of real-valued functions over the hypercube and the hypergrid has been resolved. For constant distance parameter $\eps$, it is known to be $\Theta(d\log n)$. Moreover, this bound holds for all {\em bounded-derivative} properties~\cite{CDJS15}, a large class that includes $\bb$-monotonicity and some properties quite different from monotonicity, such as the Lipschitz property. Amazingly, the upper bound for all these properties is achieved by the same simple and, in particular, nonadaptive, tester.
Even though proving lower bounds for adaptive testers has been challenging in general, a line of work, starting from Fischer~\cite{Fis04} and including \cite{BBM12,CS14,CDJS15}, has established that adaptivity does not help for this large class of properties. Since unateness is so closely related, it is natural to ask whether the same is true for testing unateness.

We answer this in the negative: we prove that any nonadaptive tester of real valued functions over the hypercube (for some constant distance parameter) must make $\Omega(d\log d)$ queries.
More generally, it needs $\Omega(d(\log d+\log n))$ queries for the hypergrid domain. These lower bounds complement our algorithms, completing the picture for unateness testing of real-valued functions.
From a property testing standpoint, our results establish that unateness is different from
monotonicity and, more generally, any derivative-bounded
property.

\subsection{Formal Statements and Technical Overview}
Our testers are summarized in the following theorem, stated for functions over the hypergrid domains. (Recall that the hypercube is a special case of the hypergrid with $n=2$.)

\begin{theorem}\label{thm:main-hg}
	Consider functions $f:[n]^d \to \R$ and a distance parameter $\eps\in (0,1/2)$.
	\begin{compactenum}
		\item\label{item:nonadaptive} There is a nonadaptive unateness tester that makes $O(\frac d \eps(\log\frac d\eps + \log n))$ queries\footnote{\scriptsize For many properties, when the domain is extended from the hypercube to the hypergrid, testers incur an extra multiplicative factor of $\log n$ in the query complexity. This is the case for our adaptive tester.
			 However, note that the complexity of nonadaptive unateness testing (for constant $\eps$) is $\Theta(d(\log d + \log n))$ rather than $\Theta(d\log d\log n).$}.
		
		\item\label{item:adaptive} There is an adaptive unateness tester  that makes $O(\frac{d\log n}\eps)$ queries.
	\end{compactenum}
	Both testers have one-sided error.
\end{theorem}

\noindent
Our main technical contribution is the proof that the extra $\Omega(\log d)$ is needed for nonadaptive testers.
This result demonstrates a gap between adaptive and nonadaptive unateness testing.
\begin{theorem}\label{thm:non-adap-lb-1}
	Any nonadaptive unateness tester (even with two-sided error) for real-valued functions $f:\{0,1\}^d \to \R$ with distance parameter $\eps = 1/8$  must make $\Omega(d\log d)$ queries.
\end{theorem}

\noindent The lower bound for adaptive testers is an easy adaptation of the monotonicity lower bound in~\cite{CS14}. 
We state this theorem for completeness and prove it in Appendix~\ref{sec:adap-lb}.
\begin{theorem}\label{thm:adap-lb}
Any unateness tester  for functions $f:[n]^d \to \R$ with distance parameter $\eps \in (0,1/4)$ must make $\Omega\left(\frac{d \log n}{\eps}-\frac{\log 1/\eps}{\eps}\right)$ queries.
\end{theorem}
Theorems~\ref{thm:non-adap-lb-1} and \ref{thm:adap-lb} directly imply that our nonadaptive tester is optimal for constant $\eps$, even for the hypergrid domain. The details appear in Appendix~\ref{sec:na-lb-hg}.

\subsubsection{Overview of Techniques}\label{sec:intro-tech}
We first consider the hypercube domain. For each $i\in[d],$ an {\em $i$-edge} of the hypercube is a pair $(x,y)$ of points in $\{0,1\}^d$, where $x_i=0,y_i=1$, and $x_j=y_j$ for all $j\in([d] \setminus \{i\})$. Given an input function $f:\{0,1\}^d\to\R$, we say an $i$-edge $(x,y)$ is {\em increasing} if $f(x)<f(y)$, {\em decreasing} if $f(x)>f(y),$ and {\em constant} if $f(x)=f(y)$.

Our nonadaptive unateness tester on the hypercube uses the work investment strategy from~\cite{BerRY14} (also refer to Section 8.2.4 of Goldreich's book~\cite{Go-book}) to ``guess'' a good dimension where to look for violations of unateness (specifically, both increasing and decreasing edges).  For all $i\in[d]$, let $\alpha_i$ be the fraction of the $i$-edges that are decreasing, $\beta_i$ be the fraction of the $i$-edges that are increasing, and $\mu_i = \min(\alpha_i,\beta_i)$. The dimension reduction theorem from~\cite{CDJS15} implies that if the input function is $\eps$-far from unate, then the average of $\mu_i$ over all dimensions is at least $\frac\eps{4d}$. If the tester knew which dimension had $\mu_i=\Omega(\eps/d)$, it could detect a violation with high probability by querying the endpoints of $O(1/\mu_i)=O(d/\eps)$ uniformly random edges.
However, the tester does not know which $\mu_i$ is large
and, intuitively, nonadaptively checks the following
$\log d$ different scenarios, one for each $k\in[\log d]$: exactly $2^k$ different
$\mu_i$'s are $\eps/2^k$, and all others are $0$. This leads to the query complexity of $O(\frac {d\log d}\eps).$

With adaptivity, this search  through $\log d$ different scenarios is not required.
A pair of queries in each dimensions detects influential coordinates (i.e., dimensions with many non-constant edges), and the algorithm focuses
on finding violations among those coordinates. This leads to the query complexity of $O(d/\eps)$, removing the $\log d$ factor.

It is relatively easy to extend (both adaptive and nonadaptive) testers from hypercubes to hypergrids by incurring an extra factor of $\log n$ in the query complexity. The role of $i$-edges is now played by {\em $i$-lines}. An $i$-line is a set of $n$ domain points that differ only on coordinate $i$. The domain $[n]$ is called a line. Monotonicity on the line (a.k.a. sortedness) can be tested with $O(\frac{\log n}\eps)$ queries, using, for example, the classical {\em tree tester} from \cite{EKKRV00}. Instead of sampling a random $i$-edge, we sample a random $i$-line $\ell$ and run the tree tester on the restriction $f_{|\ell}$ of function $f$ to the line $\ell$.
This is optimal for adaptive testers, but, interestingly, not for nonadaptive testers.
We show that for each function $f$ on the line that is $\eps$-far from unateness, one of the two scenarios happen: (1) the tree tester is likely to find a violation of unateness; (2) function $f$ is increasing (and also decreasing) on a constant fraction of pairs in $[n]$. This new angle on the classical tester allows us to replace the factor $(\log d)(\log n)$ with $\log d + \log n$ in the query complexity.
Thus, the nonadaptive complexity becomes $O(d(\log d + \log n))$, which we show is optimal.

{\bf The nonadaptive lower bound.} Our most significant finding is the $\log d$ gap in the query complexity between adaptive and nonadaptive testing of unateness.
By previous work~\cite{Fis04,CS14}, it suffices to prove lower bounds for {\em comparison-based} testers, i.e., testers that
can only perform comparisons of the function values at queried points, but cannot use the values themselves.
Our main technical contribution is the $\Omega(d\log d)$ lower
bound for nonadaptive comparison-based testers of unateness on hypercube domains.

Intuitively, we wish to construct $K=\Theta(\log d)$ families of functions where,
for each $k \in [K]$, functions in the $\ord{k}$ family have
$2^k$ dimensions $i$ with $\mu_i=\Theta(1/2^k)$, while $\mu_i=0$ for all other dimensions.
What makes the construction challenging is the existence of a \emph{single, universal} nonadaptive
$O(d)$-tester for all
$\bb$-monotonicity properties, proven in~\cite{CDJS15}. In other words, there is a single
distribution on  $O(d)$ queries that defines a nonadaptive property tester for
$\bb$-monotonicity, regardless of $\bb$.
Since unateness
is the union of all $\bb$-monotonicity properties,
our construction must be able to fool such algorithms.
Furthermore, nonadaptivity must be critical, since
we obtained a $O(d)$-query adaptive tester for unateness.

Another obstacle is that once a tester finds a non-constant edge in each dimension, the problem reduces to testing $\bb$-monotonicity for a vector $\bb$ determined by the directions (increasing or decreasing) of the non-constant edges. That is, intuitively, most edges in our construction must be constant. This is one of the main technical challenges. The previous lower bound constructions for monotonicity testing \cite{BBM12,CS14} crucially used the fact that all edges in the hard functions were non-constant.

We briefly describe how we overcome the problems mentioned above.
By Yao's minimax principle, it suffices to construct $\Yes$ and $\No$ distributions that
a deterministic nonadaptive tester cannot distinguish.
First, for some parameter $m$, we partition the hypercube into $\totcube$ subcubes based of the first $\log_2\totcube$ most significant coordinates.
Both distributions, $\Yes$ and $\No$, sample a uniform $k$ from $[K]$, where $K=\Theta(\log d)$, and a set $R\subseteq[d]$ of cardinality $2^k$.
Furthermore, each subcube $j\in [\totcube]$ selects an ``action dimension'' $r_j\in R$ uniformly at random. For both distributions, in any particular subcube $j$, the function value is completely determined by the coordinates {\em not in} $R$, and the random coordinate $r_j\in R$. Note that all the $i$-edges for $i\in (R\setminus \{r_j\})$ are constant.
Within the subcube, the function is a linear function
with exponentially increasing coefficients.
In the $\Yes$ distribution, any two cubes $j,j'$ with the same action dimension orient the edges in that dimension the same  way (both increasing or both decreasing), while in the $\No$ distribution each cube decides on the orientation independently.
The former correlation maintains unateness while the  latter independence creates distance to unateness.
We prove that to distinguish the distributions,
any comparison-based nonadaptive tester must find two distinct subcubes
with the same action dimension $r_j$ and, furthermore, make a specific
query (in both) that reveals the coefficient of $r_j$.
We show that, with $o(d\log d)$ queries, the probability of this event is negligible.

\section{Upper Bounds}\label{sec:upper-bounds}
In this section, we prove parts 1-2 of Theorem~\ref{thm:main-hg}, starting from the hypercube domain.

Recall the definition of $i$-edges and $i$-lines from Section~\ref{sec:intro-tech} and what it means for an edge to be increasing, decreasing, and constant.

The starting point for our algorithms is the dimension reduction theorem from~\cite{CDJS15}. It bounds the distance of $f:[n]^d \to \R$ to monotonicity in terms of average distances of restrictions of $f$ to one-dimensional functions.
\begin{theorem}[Dimension Reduction, Theorem 1.8 in \cite{CDJS15}] \label{thm:dimred} Fix a bit vector ${\bf b}\in\{0,1\}^d$ and a function $f:[n]^d \to \R$ which is $\eps$-far from ${\bf b}$-monotonicity.
	For all $i\in[d]$, let $\mu_i$ be the average distance of $f_{|\ell}$ to $\bb_i$-monotonicity over all $i$-lines $\ell$.
	Then, $$\sum_{i=1}^d \mu_i \geq \frac{\eps}{4}.$$
\end{theorem}
For the special case of the hypercube domains, $i$-lines become $i$-edges, and the average distance $\mu_i$ to $\bb_i$-monotonicity is the fraction of $i$-edges on which the function is not $\bb_i$-monotone.

\subsection{The Nonadaptive Tester over the Hypercube}\label{sec:non-adap-ub}
We now describe Algorithm~\ref{alg:na-unate}, the nonadaptive tester for unateness over the hypercubes.

\begin{algorithm}
\caption{The Nonadaptive Unateness Tester over Hypercubes} \label{alg:na-unate}
\SetKwInOut{Input}{input}\SetKwInOut{Output}{output}
\SetKwFor{RepeatTimes}{repeat}{times}{end}
\SetKwFor{RepeatIterations}{repeat}{iterations}{end}
\Input{distance parameter $\eps \in (0,1/2)$; query access to a function $f:\{0,1\}^d \to \R$.}
\DontPrintSemicolon
\BlankLine
\nl \For{$r = 1$ to $\cei{3\log(4d/\eps)}$}{
\nl \RepeatTimes{$s_r = \cei{\frac{16d \ln 4}{\eps\cdot 2^r}}$}{
\nl \label{step:sample-dim}Sample a dimension $i \in [d]$ uniformly at random. \;
\nl \label{step:reject-hc}Sample $3 \cdot 2^r$ $i$-edges uniformly and independently at random and {\bf reject} if there exists an increasing edge and a decreasing edge among the sampled edges.\;
}
}
\nl {\bf accept}\;
\end{algorithm}

It is evident that \Alg{na-unate} is a nonadaptive, one-sided error tester. Furthermore,
its query complexity is $O\left(\frac{d}{\eps}\log\frac{d}{\eps}\right)$. It suffices to prove
the following.

\begin{lemma} \label{lem:na-unate-rej}
If $f$ is $\eps$-far from unate,
	\Alg{na-unate} rejects with probability at least $2/3$.
\end{lemma}

\begin{proof} Recall that $\alpha_i$ is the fraction of $i$-edges that are decreasing, $\beta_i$ is the fraction of $i$-edges that are increasing and $\mu_i = \min(\alpha_i,\beta_i).$

	Define the $d$-dimensional bit vector ${\bf b}$ as follows: for each $i\in [d],$ let ${\bf b}_i = 0$
	if $\alpha_i < \beta_i$ and $\bb_i = 1$ otherwise.
	Observe that the average distance of $f$ to $\bb_i$-monotonicity over a random $i$-edge is precisely $\mu_i$.
	Since $f$ is $\eps$-far from being unate,
	$f$ is also $\eps$-far from being ${\bf b}$-monotone.
	By \Thm{dimred}, $\sum_{i \in [d]} \mu_i \geq \frac\eps 4$. Hence, $\Exp_{i \in [d]}[\mu_i] \geq \frac\eps{4d}$.
	We now apply the work investment strategy due to Berman et al.~\cite{BerRY14} to get an upper bound on the probability that \Alg{na-unate} fails to reject.

\begin{theorem}[\cite{BerRY14}]\label{thm:wis}
For a random variable $X \in [0,1]$ with $\Exp[X] \geq \mu$ for $\mu < \frac 1 2$, let $p_r = \Pr[X \geq 2^{-r}]$ and $\delta \in (0,1)$ be the desired error probability. Let $s_r = \frac{4 \ln 1/\delta}{\mu \cdot 2^r}$. Then,
$$\prod\limits_{i=1}^{\lceil 3\log (1/\mu) \rceil} (1-p_r)^{s_r} \leq \delta. $$
\end{theorem}

\noindent Consider running \Alg{na-unate} on a function $f$ that is $\eps$-far from unate. Let $X = \mu_i$ where $i$ is sampled uniformly at random from $[d]$. Then $\Exp[X] \geq \frac\eps{4d}$. Applying the work investment strategy (\Thm{wis}) on $X$ with $\mu=\frac\eps{4d}$, we get that the probability that, in some iteration, \Step{sample-dim} samples a dimension $i$ such that $\mu_i\geq 2^{-r}$ is at least $1-\delta$.
We set $\delta = 1/4$.
Conditioned on sampling such a dimension, the probability that \Step{reject-hc} fails to obtain an increasing edge and a decreasing edge among its $3 \cdot 2^r$ samples is at most $2\left(1-2^{-r}\right)^{3 \cdot 2^r} \leq 2e^{-3} < 1/9$, as the fraction of both increasing and decreasing edges in the dimension is at least $2^{-r}$.
Hence, the probability that \Alg{na-unate} rejects $f$ is at least $\frac{3}{4} \cdot \frac{8}{9} = \frac{2}{3}$,
which completes the proof of Lemma~\ref{lem:na-unate-rej}.
\end{proof}

\subsection{The Adaptive Tester over the Hypercube}\label{sec:adap-ub}
We now describe Algorithm~\ref{alg:adap-unate}, an adaptive tester for unateness over the hypercube domain with good expected query complexity. The final tester is obtained by repeating this tester and accepting if the number of queries exceeds a specified bound.

\begin{algorithm}
\caption{The Adaptive Unateness Tester over Hypercubes} \label{alg:adap-unate}
\SetKwInOut{Input}{input}\SetKwInOut{Output}{output}
\SetKwFor{RepeatTimes}{repeat}{times}{end}
\SetKwFor{RepeatIterations}{repeat}{iterations}{end}
\Input{distance parameter $\eps \in (0,1/2)$; query access to a function $f:\{0,1\}^d \to \R$.}
\DontPrintSemicolon
\BlankLine
\nl\label{step:repeat}\RepeatTimes{$10/\eps$}{
\nl\label{step:begin-adap-hc}\For{$i = 1$ to $d$}{
\nl\label{step:sample-s}Sample an $i$-edge $e_i$ uniformly at random.\;
\nl\label{step:end-build-s}\If{$e_i$ is non-constant (i.e., increasing or decreasing)}{
\nl\label{step:stage-2-edge-sample} Sample $i$-edges uniformly at random till we obtain a non-constant edge $e_i'$.\;
\nl\label{step:end-adap-hc} {\bf reject} if one of the edges $e_i, e_i'$ is increasing and the other is decreasing.\;}
}}
\nl {\bf accept}
\end{algorithm}

\begin{claim} \label{clm:time}
The expected number of queries made by
\Alg{adap-unate} is $40d/\eps$.
\end{claim}

\begin{proof}
Consider one iteration of the {\bf repeat}-loop in Step~\ref{step:repeat}.
We prove that the expected number of queries in this iteration is $4d$.
The total number of queries in Step~\ref{step:sample-s} is $2d$, as 2 points per dimension are queried.
	Let $E_i$ be the event that edge $e_i$ is non-constant and $T_i$ be the random variable for the number of $i$-edges sampled in \Step{stage-2-edge-sample}. Then $\Exp[T_i] = \frac 1{\alpha_i+\beta_i}=\frac 1{\Pr[E_i]}$. Therefore, the expected number of all edges sampled in \Step{stage-2-edge-sample} is $\sum_{i=1}^d \Pr[E_i]\cdot\Exp[T_i]=\sum_{i=1}^d \Pr[E_i] \cdot \frac 1{\Pr[E_i]} = d$. Hence, the expected number of queries in \Step{stage-2-edge-sample} is $2d$. Since there are $10/\eps$ iterations in Step~\ref{step:repeat}, the expected number of queries in \Alg{adap-unate} is $40d/\eps$.
\end{proof}

\begin{claim} \label{clm:rej}
If $f$ is $\eps$-far from unate, \Alg{adap-unate} accepts with probability at most $1/6$.
\end{claim}

\begin{proof}
First, we bound the probability that a violation of unateness is detected in some dimension $i \in [d]$ in one iteration of the {\bf repeat}-loop.
Consider the probability of finding a 
decreasing $i$-edge in Step~\ref{step:sample-s}, and an 
increasing $i$-edge in \Step{stage-2-edge-sample}.
	The former is exactly $\alpha_i,$ and the latter is $\frac{\beta_i}{\alpha_i+\beta_i}$. Therefore, the probability we detect a violation from dimension $i$ is
$\frac{2 \alpha_i \beta_i}{\alpha_i+\beta_i} \geq \min(\alpha_i,\beta_i) = \mu_i$.
The probability that we fail to detect a violation in any of the $d$ dimensions is at most $\prod_{i =1}^d (1-\mu_i) \leq \exp\big(-\sum_{i = 1}^d \mu_i\big),$ which is at most $e^{-\eps/4}$ by \Thm{dimred} (Dimension Reduction). By Taylor expansion of $e^{-\eps/4}$, the probability of finding a violation in one iteration is at least $1-e^{-\eps/4} \geq
\frac{\eps}{4} - \frac{\eps^2}{32} >
\frac{\eps}{5}$. 
The probability that \Alg{adap-unate} does not reject in any iteration is at most $(1-\eps/5)^{10/\eps} < 1/6$.
\end{proof}

\begin{proof}[{Proof of Theorem~\ref{thm:main-hg}, Part 2} (for the special case of the hypercube domain)]
We run \Alg{adap-unate}, aborting and accepting if we ever make more than $240d/\eps$ queries. By Markov's inequality, the probability of aborting is at most $1/6$.
By Claim~\ref{clm:rej}, if $f$ is $\eps$-far from unate, \Alg{adap-unate} accepts with probability at most $1/6$. The theorem follows by a union bound.
\end{proof}

\subsection{Extension to Hypergrids}\label{sec:hypergrids}
We start by establishing terminology for lines and pairs.
Consider a function $f:[n]^d\to\R$.
Recall the definition of $i$-lines from Section~\ref{sec:intro-tech}.
A pair of points that differ only in coordinate $i$ is called an $i$-pair.
An $i$-pair $(x,y)$ with $x_i < y_i$ is called {\em increasing} if $f(x) < f(y)$, {\em decreasing} if $f(x) > f(y),$ and {\em constant} if $f(x) = f(y)$.

\begin{algorithm}
\caption{Tree Tester}\label{alg:tree-tester}
\SetKwInOut{Input}{input}\SetKwInOut{Output}{output}
\SetKwFor{RepeatTimes}{repeat}{times}{end}
\SetKwFor{RepeatIterations}{repeat}{iterations}{end}
\Input{Query access to a function $h:[n] \mapsto \R$.}
\DontPrintSemicolon
\BlankLine

\nl Pick $x \in [n]$ uniformly at random. \;
\nl Let $Q_x\subseteq [n]$ be the set of points visited in a binary search for $x$. Query $h$ on all points in $Q_x$. \;
\nl If there is an increasing pair in $Q_x$, set $\dir \gets \{\uparrow\}$; otherwise, $\dir\gets\emptyset.$\;
\nl If there is a decreasing pair in $Q_x$, update $\dir \gets\dir \cup \{\downarrow\}$. \;
\nl {\bf Return} $\dir$.
\end{algorithm}

The main tool for extending Algorithms~\ref{alg:na-unate} and \ref{alg:adap-unate} to work on hypergrids is the {\em tree tester}, designed by Ergun et al.~\cite{EKKRV00} to test monotonicity of functions $h:[n]\to\R$.
We modify the tree tester to return information about directions it observed instead of just accepting or rejecting. See Algorithm~\ref{alg:tree-tester}. The following lemma summarizes the guarantee of the tree tester.

\begin{lemma}[\cite{EKKRV00, CDJS15}] \label{lem:line}
If $h:[n] \mapsto \R$ is $\eps$-far from monotone (respectively, antimonotone), then the output of \Alg{tree-tester} on $h$ contains $\downarrow$ (respectively, $\uparrow$) with probability at least $\eps$.
\end{lemma}

\begin{algorithm}
\caption{The Adaptive Unateness Tester over Hypergrids} \label{alg:adap-unate-hg}
\SetKwInOut{Input}{input}\SetKwInOut{Output}{output}
\SetKwFor{RepeatTimes}{repeat}{times}{end}
\SetKwFor{RepeatIterations}{repeat}{iterations}{end}
\Input{distance parameter $\eps \in (0,1/2)$; query access to a function $f:[n]^d \to \R$.}
\DontPrintSemicolon
\BlankLine
\nl \RepeatTimes{$10/\eps$}{
\nl\label{step:begin-adap-hg}\For{$i = 1$ to $d$}{
\nl Sample an $i$-line $\ell_i$ uniformly at random.\;
\nl\label{step:test-i-line}Let $\dir_i$ be the output of \Alg{tree-tester} on $f_{|\ell_i}$. \;
\nl \If{$\dir_i \ne \emptyset$}{
\nl\label{step:verify-i-line} Sample $i$-lines uniformly at random and run \Alg{tree-tester} on $f$ restricted to each line until it returns a non-empty set. Call it $\dir_i'$.\;
\nl  If $\dir_i \cup \dir_i' = \{\uparrow, \downarrow \}$, {\bf reject}.\;
}}}
\nl {\bf accept}
\end{algorithm}

\begin{algorithm}[h!]
	\caption{The Nonadaptive Unateness Tester over Hypergrids} \label{alg:na-unate-hg}
	\SetKwInOut{Input}{input}\SetKwInOut{Output}{output}
	\SetKwFor{RepeatTimes}{repeat}{times}{end}
	\SetKwFor{RepeatIterations}{repeat}{iterations}{end}
	\Input{distance parameter $\eps \in (0,1/2)$; query access to a function $f:[n]^d \to \R$.}
	\DontPrintSemicolon
	\BlankLine
	
	\nl \label{step:mono}\RepeatTimes{$220/\eps$}{
		\nl \label{step:mono-start}\For{$i=1$ to $d$}{
			\nl Sample an $i$-line $\ell$ uniformly at random. \;
			\nl \label{step:log}{\bf Reject} if \Alg{tree-tester}, on input $f_{|\ell}$, returns $\{\uparrow,\downarrow\}$.
		}
	}
	\nl \label{step:cube}\For{$r=1$ to $\cei{3 \log(200d/\eps)}$}{
		\nl \RepeatTimes{$s_r = \cei{\frac{800d \ln 4}{\eps \cdot 2^r}}$}{
			\nl \label{step:samp-hg1}Sample a dimension $i \in [d]$ uniformly at random.\;
			\nl \label{step:tree-tester1}Sample $3 \cdot 2^r$ $i$-pairs uniformly and independently at random.\;
			\nl \label{step:rej-hg1}If we find an increasing and a decreasing pair among the sampled pairs, {\bf reject}.\;
		}
	}
	\nl {\bf accept}			
\end{algorithm}

Our hypergrid testers are stated in Algorithms~\ref{alg:adap-unate-hg} and~\ref{alg:na-unate-hg}.
Next, we explain how Lemma~\ref{lem:line} and Theorem~\ref{thm:dimred} are used in the analysis of the adaptive tester.
For a dimension $i\in[d]$, let $\alpha_i$ and $\beta_i$ denote the average distance of $f_{|\ell}$ to monotonicity and antimonotonicity, respectively, over all $i$-lines $\ell$.
Then $\mu_i := \min(\alpha_i,\beta_i)$ is the average fraction of points per $i$-line that needs to change to make $f$ unate.
Define the $\bb$-vector with $\bb_i = 0$ if $\alpha_i < \beta_i$,
and $\bb_i = 1$ otherwise.
By \Thm{dimred},
if $f$ is $\eps$-far from unate, and thus $\eps$-far from $\bb$-monotone, then $\sum_{i=1}^d \mu_i \geq \eps/4$.
By \Lem{line}, the probability that the output of \Alg{tree-tester} on $f_{|\ell}$
contains $\downarrow$ (respectively, $\uparrow$), where $\ell$ is a uniformly random $i$-line,
is at least $\alpha_i$ (respectively, $\beta_i$).
The rest of the analysis of Algorithm~\ref{alg:adap-unate-hg} is similar to that in the hypercube case.

\begin{proof}[Proof of Theorem~\ref{thm:main-hg}, Part 2]
	The tester is Algorithm~\ref{alg:adap-unate-hg}.
As in the proof of Claim~\ref{clm:time}, the expected running time of Algorithm~\ref{alg:adap-unate-hg} is at most $(40d\log n)/\eps$.
The proof of Claim~\ref{clm:rej} carries over almost word-to-word. Fix dimension $i$. The probability that $\downarrow \in \dir_i$ in Step~\ref{step:test-i-line} is at least $\alpha_i$.
The probability that $\uparrow \in \dir_i'$ in Step~\ref{step:verify-i-line} is at least $\frac{\beta_i}{\alpha_i + \beta_i}$.
The rest of the calculation is identical
to that of the proof of \Clm{rej}.
\end{proof}

To analyze the nonadaptive tester,
we prove Lemma~\ref{lem:treetester}, which demonstrates the power of the tree tester and  may be of independent interest.

\begin{lemma}\label{lem:treetester}
	Consider a function $h:[n]\to \R$ which is $\eps$-far from monotone (respectively, antimonotone). At least one of the following holds:
	\begin{compactenum}
		\item $\Pr[\text{\Alg{tree-tester}, on input $h$, returns }\{\uparrow,\downarrow\}] \geq \eps/25$.
		\item $\Pr_{u,v \in [n]}[(u,v) \textrm{ is a decreasing (respectively, increasing)  pair}] \geq \eps/25$.
	\end{compactenum}
\end{lemma}

\begin{proof}
\def\lca{\mathsf{lca}}
\def\updown{$\{\uparrow, \downarrow \}$}
Let $T$ be a balanced binary search tree consisting of elements in $[n]$, such that the set of points visited in a binary search for some $x \in [n]$ corresponds to a path from the root to the node containing $x$ in $T$.
Let $Q_x$ denote the set of points visited in a binary search for $x \in [n]$.
For $x, y \in [n]$, denote the least common ancestor of $x$ and $y$ by $\lca(x,y)$.

Let $W_{\uparrow \downarrow}$ be a set of points $x$ such that $Q_x$ contains both an increasing and a decreasing pair (with respect to $h$). If $|W_{\uparrow \downarrow}|\geq \frac{\eps n}{10}$, then Case 1 of Lemma~\ref{lem:treetester} holds. We may therefore assume that $|W_{\uparrow \downarrow}|< \frac{\eps n}{10}$.
Let $\calE$ be the event that for any $u,v \in [n]$ such that $u < v$, the pair $(u,v)$ is decreasing. We will prove
that $\Pr[\calE] \geq \eps/25$.

Let $W_{\downarrow}$ be that set of points $x \in [n]$ such that $Q_x$
contains a decreasing pair. Similarly, define the set $W_{\uparrow}$. Let $W_c$ denote the set of points $x$ such that $h_{|Q_x}$ is constant.

\begin{claim}[\cite{EKKRV00}]\label{clm:2}
	The function $h$ restricted to the set $W_{\uparrow}\cup W_c$ is monotone.
\end{claim}

\begin{proof}
The proof is by contradiction. Suppose $x,y \in (W_{\uparrow} \cup W_c)$ such that $x < y$, but $h(x) > h(y)$.
Consider $z = \lca(x,y)$. Either $h(x) > h(z)$ or $h(z) > h(y)$, contradicting
the fact that $x, y \in W_{\uparrow}\cup W_c$.
\end{proof}

\noindent By symmetry, the function $h$ restricted to the set $W_{\downarrow}\cup W_c$ is antimonotone.

A priori, points in $W_{\uparrow}$ and $W_{\downarrow}$ could be interspersed.
The next claim shows that they are in different halves of the tree $T$.

\begin{claim}\label{clm:pure}
	If $x\in W_{\downarrow}$ and $y\in W_{\uparrow}$, then $\lca(x,y)$ is the root of $T$ (which is equal to $\lceil n/2 \rceil$).
\end{claim}

\begin{proof}
	Suppose not. Let $z := \lca(x,y)$ and $w$ be the parent of $z$.
    Consider the case where $z$ is the left child of $w$,  $x$ lies in the left subtree of $z$ and $y$ lies in the right subtree of $z$.
    (All the other cases have analogous proofs.)
    Observe that all points in $Q_y$ lie in the interval $[z,w]$.
    Both $w$ and $z$ are in $Q_x$ as well as in $Q_y$.
    As $x \in W_{\uparrow}$ and $y \in W_{\downarrow}$, it must be the case that $h(w) = h(z)$.
    Since $y\notin W_{\uparrow \downarrow}$, for all $p \in Q_y$, we have $h(p) = h(w)$.
    This contradicts the fact that $y\in  W_{\uparrow}$.

    In all cases, we conclude that either $x\notin W_{\downarrow}$ or
    $y\notin W_{\uparrow}$.
    Thus, $z$ cannot have a parent, and $z = \lceil n/2 \rceil$.
\end{proof}

\begin{claim}\label{clm:obv}
Let $g:[n] \mapsto \R$ be an antimonotone function and $\dist(g, \mathrm{constant})$ denote the fraction of points that need to be changed so that $g$ is a constant function.
	If $g$ is antimonotone, and $\dist(g, \mathrm{constant}) \geq \rho$, where $\rho \le \frac 1 2$, then
	$$\Pr\limits_{u,v \in [n]: u < v} [(u,v) \textrm{ is decreasing}] \geq \frac\rho 2 .$$
\end{claim}

\begin{proof}
	The probability that $g(u) \neq g(v)$ is at least $\rho(1-\rho)$ which is at least $\frac \rho 2$ when $\rho \le \frac 1 2$. Since $g$ is antimonotone, $(u,v)$ is a decreasing pair.
\end{proof}

\noindent
Let $L$ (respectively, $R$) be the set of points in $[n] \setminus W_{\uparrow \downarrow}$ in the left (respectively, right) subtree of the root. Define $\mu_L := |L|/n$; similarly, define $\mu_R$.
Observe that both $\mu_L$ and $\mu_R$ are at least $\frac{1}{2} - \frac{\eps}{10}$.
By Claims~\ref{clm:2} and~\ref{clm:pure}, $h_{|L}$ (and $h_{|R}$) is either monotone or antimonotone.
Now, if any of these two functions were antimonotone and $\frac\eps 2$-far from being constant (w.l.o.g., assume $h_{|L}$ satisfies the condition), then by \Clm{obv}, we would have
\[
\Pr[\calE] \geq \Pr\limits_{u < v}\left[(u,v) \textrm{ is decreasing and }  u,v \in L \right] \geq  \frac\eps 4 \cdot \left(\frac{1}{2} - \frac{\eps}{10} \right)^2  \geq \frac{\eps}{25}.
\]
Assume that this doesn't occur. We have two cases.

\noindent {\bf Case 1.} Both $h_{|L}$ and $h_{|R}$ are $\frac\eps 2$-close\footnote{A function $h$ is $\eps$-close to a property $\cP$ if it is sufficient to change at most $\eps$-fraction of values in $h$ to make it satisfy $\cP$.}
to being constant. In this case, at least $(1-\frac\eps 2)|L|$ points of $L$ evaluate to a constant $C_1$, and at least $(1-\frac\eps 2)|R|$ points of $R$ evaluate to constant $C_2$. We must have $C_1 > C_2$, for otherwise, we can make $h$ monotone by changing only $\frac\eps 2 \cdot (|R|+|L|)+\frac{\eps n}{10} < \eps n$ points, which is a contradiction.
Hence,
\[
\Pr[\calE] \geq \Pr\limits_{u < v} \left[h(u) = C_1 \textrm{ and } h(v) = C_2 \right] \geq \left(1-\frac\eps 2 \right)^2 \mu_L\mu_R > \frac{1}{4} \cdot \left(\frac 1 2- \frac{\eps}{10} \right)^2 \geq \frac{\eps}{25}.
\]

\noindent {\bf Case 2.} At least one of the functions is $\frac\eps 2$-far from being constant and is monotone.
W.l.o.g., assume $h_{|L}$ satisfies this condition.
Note that all points in $L$ are only in $W_{\uparrow}\cup W_c$, and so, all points in $R$ must be in $W_{\downarrow}\cup W_c$.
This implies that $h_{|R}$ is antimonotone. (Note that a constant function is also antimonotone.)
But then, $h_{|R}$ must be $\frac \eps 2$-close to being constant.
Then at least $(1-\frac\eps 2)|R|$ points in $R$ evaluate to a constant, say $C$.
Let $U$ denote the set of points in $L$ whose values are strictly
greater than $C$. Since $h_{|L}$ is monotone, we can make
$h$ monotone by deleting all points in $U, W_{\uparrow \downarrow}$,
and the points in $R$ that do not evaluate to $C$.
The total number of points to be deleted is at most $|U| + \frac{\eps n}{10} + \frac{\eps n}{2}$, which must be at least $\eps n$, as $h$ is $\eps$-far from monotone. Hence, $|U| > \eps n/3$.
Therefore,
\[
\Pr[\calE] \geq \Pr\limits_{u < v} \left[ u \in U \textrm{ and } h(v) = C \right] \geq \frac{\eps}{3} \cdot  \left(1 - \frac\eps 2 \right)\mu_R  > \frac{\eps}{25} .
\]
This completes the proof of \Lem{treetester}.
\end{proof}

We now analyze \Alg{na-unate-hg}.
It is evident that it has one-sided error and makes
$O(\frac d \eps(\log n + \log\frac d\eps))$ queries.
It suffices to prove the following.

\begin{theorem} \label{thm:non-adap-hyper} If $f:[n]^d \mapsto \R$ is $\eps$-far
	from unate, then \Alg{na-unate-hg} rejects with probability at least $2/3$.
\end{theorem}

\begin{proof} For any line $\ell$, we define the following quantities.
\begin{compactitem}
\item $\alpha_\ell$: the distance of $f_{|\ell}$ to monotonicity.
\item $\beta_\ell$: the distance of $f_{|\ell}$ to antimonotonicity.
\item $\sigma_\ell$: the probability that \Alg{tree-tester}, on input $f_{|\ell}$, returns $\{\uparrow, \downarrow\}$.
\item $\delta_\ell$: the probability that a uniformly random pair in $\ell$ is decreasing.
\item $\lambda_\ell$: the probability that a uniformly random pair in $\ell$ is increasing.
\end{compactitem}

\noindent Let $L_i$ be the set of $i$-lines. By 
\Thm{dimred},
	$$\frac{1}{n^{d-1}} \sum_{i=1}^d \min\left(\sum_{\ell \in L_i} \alpha_\ell, \sum_{\ell \in L_i} \beta_\ell \right) \geq \frac{\eps}{4}.$$
	By \Lem{treetester}, for every line $\ell$, we have
	$\sigma_\ell + \delta_\ell \geq \alpha_\ell/25$ and $\sigma_\ell + \lambda_\ell \geq \beta_\ell/25$.
	Also note,
	$$
	\frac{1}{n^{d-1}}\sum_{i=1}^d \left[\sum_{\ell \in L_i} \sigma_\ell + \min\left(\sum_{\ell \in L_i} \delta_\ell, \sum_{\ell \in L_i} \lambda_\ell \right)\right] \geq \frac{1}{n^{d-1}}\sum_{i=1}^d  \min\left(\sum_{\ell \in L_i} (\sigma_\ell + \delta_\ell), \sum_{\ell \in L_i} (\sigma_\ell + \lambda_\ell) \right)$$
	Combining these bounds, we obtain that the LHS is at least $\eps/100$.
	%
	%
	Note that the first term,
	which is equal to $\sum_{i=1}^d \EX_{\ell \in L_i} [\sigma_\ell]$,
	is the expected
	number of times a single iteration of Steps~\ref{step:mono-start}-\ref{step:log}
	rejects. If this quantity is at least $\eps/200$, then
	the tester rejects with probability at least $2/3$. If not,
	then we have $n^{-(d-1)} \sum_{i=1}^d \min(\sum_{\ell \in L_i} \delta_\ell, \sum_{\ell \in L_i} \lambda_\ell) \geq \eps/200$.
	Using a calculation identical to that of the proof of Lemma~\ref{lem:na-unate-rej},
	the probability that Step~\ref{step:rej-hg1} rejects in some iteration is at least
	$2/3$.
\end{proof}

\section{The Lower Bound for Nonadaptive Testers over Hypercubes} \label{sec:lb}
In this section, we prove Theorem~\ref{thm:non-adap-lb-1}, which gives a lower bound for nonadaptive unateness testers for functions over the hypercube.

Previous work of~\cite{CS14} on lower bounds for monotonicity testing shows that,
for a special class of properties, which includes unateness,
it is sufficient to prove lower bounds for {\em comparison-based testers}.
Comparison-based testers base their decisions only on the {\em order} of the function values  at queried points,  and not on the values themselves.

We first state the reduction to comparison-based testers from~\cite{CS14}. Let a $(t,\eps,\delta)$-tester for a property $\cP$
be a $t$-query tester, with distance parameter $\eps$, that errs with (two-sided) probability at most $\delta$.
Consider functions of the form $f:D \to \R$,
where $D$ is an arbitrary
partial order (in particular the hypergrid/cube).
A property $\cP$ is invariant under monotone transformations if,
for all strictly increasing maps $\phi: \R \to \R$ and all functions $f$, $\dist(f,\cP) = \dist(\phi\circ f, \cP)$.
In particular, unateness is invariant under monotone transformations.

\begin{theorem}[implicit in Theorem 2.1 of~\cite{CS14}]\label{thm:CS14} Let $\cP$ be a property invariant under monotone
transformations. Suppose there exists a nonadaptive (resp., adaptive) $(t,\eps,\delta)$-tester for $\cP$.
Then there exists a  nonadaptive (resp., adaptive) comparison-based $(t,\eps,2\delta)$-tester
for $\cP$.
\end{theorem}

\noindent Our main lower bound theorem is stated next. In the light of the previous discussion, it implies Theorem~\ref{thm:non-adap-lb-1}.

\begin{theorem}\label{thm:non-adap-lb}
Any nonadaptive comparison-based tester for unateness of functions $f:\{0,1\}^d\to \R$ must make $\Omega(d\log d)$ queries.
\end{theorem}

\noindent By \Thm{CS14} and Yao's minimax principle~\cite{Yao77}, it suffices to prove
the lower bound for deterministic, nonadaptive, comparison-based testers over a known distribution of functions.
It may be useful for the reader to recall the  sketch of the main ideas given in Section~\ref{sec:intro-tech}.
For convenience, assume $d$ is a power of $2$ and let $d' := d+\log_2d$.
We will focus on functions $h:\{0,1\}^{d'} \to \R$,
and prove the lower bound of $\Omega(d \log d)$ for this class of functions,
as $\Omega(d \log d) = \Omega(d' \log d')$.

\subsection{The Hard Distributions}
We first partition $\{0,1\}^\dd$ into $d$ subcubes based on the most significant $\log_2 d$ bits.
More precisely, for $i \in [d]$, the $\ord{i}$ subcube is defined as
\[C_i := \{x\in \{0,1\}^\dd: \dec(x_{d'}x_{d'-1}\cdots x_{d+1}) = i - 1\},\]
where $\dec(z)$ denotes the integer equivalent to the binary string $z$. Specifically, $\dec(z_p z_{p-1} \ldots z_1) = \sum_{i = 1}^p z_i 2^{i-1}$.

Let $m=d$. We denote the set of indices of the subcube by $[m]$ and the set of dimensions by $[d]$.
We use $i,j\in [m]$ to index subcubes,
and $a,b\in [d]$ to index dimensions.
We now define a series of random variables, where each subsequent variable may depend on the previous ones.
\begin{compactitem}
    \item $k$: a number picked uniformly at random from $\left[\frac{1}{2}\log_2 d \right]$.
    \item $R$: a uniformly random subset of $[d]$ of size $2^k$.
    \item $r_i$: for each $i \in [m]$, $r_i$ is picked from $R$ uniformly and independently at random.
    \item $\alpha_b$: for each $b \in [d]$, $\alpha_b$ is picked from $\{-1,+1\}$ uniformly and independently at random. (Note: $\alpha_b$ only needs to be defined for each $b \in R$. We define it over $[d]$ just so that it is independent of $R$.)
    \item $\beta_i$: for each $i \in [m]$, $\beta_i$ is picked from $\{-1,+1\}$ uniformly and independently at random.
\end{compactitem}

\noindent We denote by $\bS$ the tuple $(k,R,\{r_i\})$, also referred
to as the \emph{shared randomness}. We use $\bT$ to refer
to the entire set of random variables.
Given $\bT$, define the following functions:
\begin{align*}
f_{\bT}(x) &:=  \sum_{b \in [d'] \setminus R} x_b 3^b + \alpha_{r_i} \cdot x_{r_i}3^{r_i},\textrm{ where $i$ is the subcube with } x \in C_i.\\
g_{\bT}(x) &:=  \sum_{b \in [d'] \setminus R} x_b 3^b + \beta_i \cdot x_{r_i}3^{r_i}, \textrm{ where $i$ is the subcube with } x \in C_i.
\end{align*}
The distribution $\Yes$ generates $f_{\bT}$ and the distribution $\No$ generates $g_{\bT}$.

In all cases, the function restricted to any subcube $C_i$ is linear.
Consider some dimension $b \in R$. There can be numerous
$r_i$'s that are equal to $b$. For $f_{\bT}$, in all of these subcubes,
the coefficient of $x_{r_i}$ has the same sign, namely $\alpha_{r_i}$. 
For $g_{\bT}$, the coefficient $\beta_i$ is potentially different,
as it depends on the actual subcube. 

\begin{claim}\label{clm:yes}
Every $f \in \supp(\Yes)$ is unate.
\end{claim}

\begin{proof}
Fix some $f \in \supp(\Yes)$. Since $f$ restricted to any $C_i$ is linear,
it suffices to argue that the coefficient of any $x_b$ (when it is non-zero) has the same sign, in all $C_i$'s.
For any $b \in [d'] \setminus R$, the coefficient of $x_b$ is always $3^b$. 
If $b \in R$, then the coefficient is either $0$ or $3^b\alpha_b$. 
\end{proof}

\begin{claim}\label{clm:no}
A function $g \sim \No$ is $\frac{1}{8}$-far from unate with probability at least $ 9/10$.
\end{claim}

\begin{proof}
	Fix $\bT = \cT$. Condition on any choice of $k$ and $R$. Note that $|R| \leq \sqrt{d}$.
	For any $r \in R$, let $A_r = \{ i : r_i = r \}$ denote the set of subcube indices with $r_i = r$.	
	Observe that $\EX[|A_r|] \geq m/\sqrt{d} = \sqrt{d}$. By a Chernoff bound and union bound, for all $r \in R$,
	we have $|A_r| \geq \sqrt{d}/2$	with probability at least $1 - d\exp(-\sqrt{d}/8)$.
	
	Condition on the event that $|A_r| \geq \sqrt{d}/2$ for all $r \in R$.
	For each $i \in A_r$, there is a random choice of $\beta_i$.
	Partition $A_r$ into $A^+_r$ and $A^-_r$,
	depending on whether $\beta_i$ is $+1$ or $-1$. Again, by a Chernoff bound
	and union bound, for all $r \in R$, we have $\min(|A^+_r|,|A^-_r|) \geq |A_r|/4$ with probability at least $1-d\exp(-\sqrt{d}/32)$.
	Thus, we can assume the above event holds with probability at least $1-d(\exp(-\sqrt{d}/8)+\exp(-\sqrt{d}/32))$, which is at least $9/10$, for large enough $d$ and
	for any choice of $k$ and $R$. 
	
	Denote the size of any subcube $C_i$ by $s$.
	In $g_{\cT}$, for all $i \in A^+_r$, all $r$-edges in $C_i$ are increasing, whereas, for all $j \in A^-_r$, all $r$-edges in $C_j$ are decreasing. 
To make $g_{\cT}$ unate, we must make all these edges have the same direction (i.e., increasing or decreasing).
	This requires modifying at least $\frac{s}{2} \cdot \min(|A^+_r|,|A^-_r|) \geq \frac{s|A_r|}{8}$ values in $g_{\cT}$. Summing over all $r$, we need to change at least $\frac{s}{8}\sum_r |A_r|$
	values. Since the $A_r$'s partition the set of subcubes, this corresponds to at least a $\frac{1}{8}$-fraction of the domain.
\end{proof}

\subsection{From Functions to Signed Graphs that are Hard to Distinguish}
For convenience, denote $x \prec y$ if $\dec(x) < \dec(y)$. Note that $\prec$ forms a total ordering
on $\{0,1\}^{\dd}$.
Given $x \prec y\in \{0,1\}^\dd$ and a function $h:\{0,1\}^\dd\to \R$, define
$\sgn_h(x,y)$ to be $1$ if $h(x) < h(y)$, $0$ if $h(x) = h(y)$,
and $-1$ if $h(x) > h(y)$.

Any deterministic, nonadaptive, comparison-based tester is defined as follows:
It makes a set of queries $Q$ and decides whether or not the input function $h$ is unate
depending on the $|Q|\choose{2}$-comparisons in $Q$.
More precisely, for every pair $(x,y) \in Q \times Q$,
$x \prec y$, we insert
an edge labeled with $\sgn_h(x,y)$. Let this signed graph be called $G^Q_h$.
Any nonadaptive, comparison-based algorithm can be described
as a partition of the universe of all signed graphs over $Q$  into $\cG_Y$ and $\cG_N$.
The algorithm accepts the function $h$ iff $G^Q_h \in \cG_Y$.

Let $\bG^Q_Y$ be the distribution of the signed graphs $G^Q_h$ when $h\sim \Yes$. Similarly, define $\bG^Q_N$ when $h\sim \No$. Our main technical theorem is \Thm{tv}, which is proved in \Sec{thm-tv-proof}.
\begin{theorem}\label{thm:tv}
For small enough $\delta > 0$ and large
enough $d$, if $|Q| \leq \delta d\log d$, then $\|\bG^Q_Y - \bG^Q_N\|_{\mathrm{TV}} = O(\delta)$.
\end{theorem}

\noindent We now prove that \Thm{tv} implies \Thm{non-adap-lb}, the main lower bound.

\begin{proof}[{\bf Proof of Theorem~\ref{thm:non-adap-lb}}]
	Consider the distribution over functions where with probability $1/2$, we sample from $\Yes$ and with the remaining probability we sample from $\No$.
	By \Thm{CS14} and Yao's minimax principle,
it suffices to prove that any deterministic, nonadaptive, comparison-based tester making
at most $\delta d\log d$ queries (for small enough $\delta > 0$) errs with probability at least $ 1/3$. Now, note that
\begin{align*}
\Pr[\textrm{error}] = \frac{1}{2} \cdot \Pr_{h\sim \Yes} [G^Q_h \in \cG_N] 
	+ \frac{1}{2} \cdot \Pr_{h\sim \No} [G^Q_h \in \cG_Y \text{ and } h \ \textrm{is $\frac{1}{8}$-far from unate}].
\end{align*}

\noindent By \Thm{tv}, the first term is at least $\frac{1}{2}\cdot \left(\Pr_{h\sim \No} [G^Q_h \in \cG_N]  - O(\delta) \right)$, and by \Clm{no},
	the second term is at least $\frac{1}{2}\cdot \left(\Pr_{h\sim \Yes} [G^Q_h \in \cG_Y] -O(\delta) - \frac{1}{10}\right)$. Summing them up, we get
$\Pr[\textrm{error}] \geq \frac{1}{2} - O(\delta) - \frac{1}{20}$ which is at least $\frac 1 3$ for small enough $\delta$.
\end{proof}

The proof of \Thm{tv} is naturally tied to the behavior of $\sgn_h$.
Ideally, we would like to say that $\sgn_h(x,y)$ is almost identical
regardless of whether $h \sim \Yes$ or $h \sim \No$. Towards this,
we determine exactly the set of pairs $(x,y)$ that potentially
differentiate $\Yes$ and $\No$.

\begin{claim}\label{clm:inv}
For all $h \in \supp(\Yes) \cup \supp(\No)$, for all $x \in C_i$ and $y \in C_j$ such that $i < j$, we have $\sgn_h(x,y) = 1$.
\end{claim}

\begin{proof}
	For any $h$, we can write $h(x)$ as $\sum_{b > d} 3^b \cdot x_b + \sum_{b \leq d} c_b(x) \cdot 3^b \cdot x_b$,
	where $c_b: \{0,1\}^{d'} \to \{-1,0,+1\}$. Thus,
	$h(y) - h(x) = \sum_{b > d} 3^b(y_b - x_b) + \sum_{b \leq d} 3^b(c_b(y) \cdot y_b - c_b(x) \cdot x_b)$.
	Recall that $x \in C_i , y \in C_j$, and $j > i$. Let $q$ denote
	the most significant bit of difference between $x$ and $y$. We have
	$q > d$, and $y_q = 1$ and $x_q = 0$. Note that for $b \leq d$, $|c_b(y) \cdot y_b - c_b(x) \cdot x_b)| \leq 2$.
	Thus, $h(y) - h(x) \geq 3^q - 2\sum_{b < q} 3^b > 0$.
\end{proof}

\noindent Thus, comparisons between points in different subcubes reveal no information
about which distribution $h$ was generated from. Therefore, the ``interesting'' pairs that can distinguish whether $h \sim \Yes$ or $h \sim \No$ must lie in the same subcube.
The next claim shows a further criterion that is needed for a pair to be interesting.
We first define another notation.

\begin{definition} \label{def:t} 
For any setting of the shared randomness $\bS$,
subcube $C_i$, and points $x,y \in C_i$, we define $\coord{i}{\bS}(x,y)$
to be the most significant coordinate of difference (between $x,y$)
in $([d] \setminus R) \cup \{r_i\}$.
\end{definition}

\noindent Note that $\bS$ determines $R$ and $\{r_i\}$.
For any $\bT$ that extends $\bS$ and any function,
the restriction to $C_i$ is unaffected by the coordinates in $R \setminus r_i$.
Thus, $\coord{i}{\bS}(x,y)$ is the first coordinate of difference that is influential
in $C_i$.

\begin{claim}\label{clm:interesting}
Fix some $\bS$, subcube $C_i$, and points $x,y \in C_i$.
Let $c = \coord{i}{\bS}(x,y)$, and assume $x \prec y$.
For any $\bT$ that extends $\bS$:
\begin{compactitem}
    \item If $c \neq r_i$, then $\sgn_{f_{\bT}}(x,y) = \sgn_{g_{\bT}}(x,y) = 1$.
    \item If $c = r_i$, $\sgn_{f_{\bT}}(x,y) = \alpha_{c}$ and $\sgn_{g_{\bT}}(x,y) = \beta_i$.
\end{compactitem}
\end{claim}

\begin{proof}
Assume $x \in C_i$.
	Recall that $f_{\bT}(x) = \sum_{b \in [d'] \setminus R} x_b 3^b + \alpha_{r_i}\cdot x_{r_i} 3^{r_i}$ and
	$g_{\bT}(x) = \sum_{b \in [d'] \setminus R} x_b 3^b + \beta_i \cdot x_{r_i} 3^{r_i}$.
	
	First, consider the case $c \neq r_i$. Thus, $c \notin R$.
	Observe that $x_b = y_b$, for all $b > c$ such that $b \notin R$.
	Furthermore, $x_c = 0$ and $y_c = 1$. Thus,
	$f_{\bT}(y) - f_{\bT}(x) > 3^c - \sum_{b < c} 3^b > 0$.
	An identical argument holds for $g_{\bT}$.
	
	Now, consider the case $c = r_i$. Thus,
	$f_{\bT}(y) - f_{\bT}(x) = \alpha_c 3^c + \sum_{b < c, b \notin R}
	(y_b - x_b) 3^b$. Using the same geometric series arguments
	as above, $\sgn_{f_{\bT}}(x,y) = \alpha_c$.
	An analogous argument shows that $\sgn_{g_{\bT}}(x,y) = \beta_i$
\end{proof}

\subsection{Proving \Thm{tv}: Good and Bad Events} \label{sec:thm-tv-proof}
For a given $Q$, we first identify certain ``bad'' values for $\bS$,
on which $Q$ could potentially distinguish between $f_{\bS}$ and $g_{\bS}$.
We will
prove that the probability
of a bad $\bS$ is small for a given $Q$.
Furthermore, we show that $Q$ cannot distinguish
between $f_{\bS}$ and $g_{\bS}$ for any good $\bS$.
We set up some definitions.

\begin{definition} \label{def:cap}
	Given a pair $(x,y)$, define $\capt(x,y)$ to be the 5 most significant coordinates\footnote{There is nothing special about the constant $5$. It just needs to be sufficiently large.}
	in which they differ.
We say $(x,y)$ {\em captures} these coordinates.
For any set $S\subseteq \{0,1\}^{\dd}$, define $\capt(S) := \bigcup_{x,y \in S} \capt(x,y)$
to be the coordinates captured by the set $S$.
\end{definition}

\noindent Fix any $Q$. We set $Q_i := Q \cap C_i$.
We define two bad events for $\bS$.
\begin{itemize}
	\item Abort Event $\cA$: There exists $x,y \in Q$ with
    $\capt(x,y) \subseteq R$.
	\item Collision Event $\cC$: There exists $i,j \in [d]$
     with $r_i = r_j$, $r_i\in \capt(Q_i)$ and $r_j\in \capt(Q_j)$.
\end{itemize}

\noindent If the abort event doesn't occur, then for any pair $(x,y)$, the sign $\sgn_h(x,y)$ is determined by $\capt(x,y)$ for any $h\in \supp(\Yes)\cup \supp(\No)$.
The heart of the analysis lies in \Thm{bad}, which states that the bad events happen rarely.
\Thm{bad} is proved in~\Sec{bad}.

\begin{theorem} \label{thm:bad} If $|Q| \leq \delta d\log d$,
then $\Pr[\cA \cup \cC] = O(\delta)$.
\end{theorem}

\noindent When neither the abort nor the collision events happen,
we say $\bS$ is good for $Q$.
Next, we show that conditioned on a good $\bS$,
the set $Q$ cannot distinguish $f \sim \Yes$ from $g \sim \No$.

\begin{lemma}\label{lem:zero}
For any signed graph $G$ over $Q$,
$$\Pr_{f\sim \Yes} [G^Q_f = G|\bS \textrm{ is good}] \!=\! \Pr_{g\sim \No} [G^Q_g = G|\bS \textrm{ is good}].$$
\end{lemma}
\begin{proof}
We first describe the high level ideas in the proof. As stated above, when the abort event doesn't happen, the sign $\sgn_h(x,y)$ is determined by $\capt(x,y)$ for any $h\in \supp(\Yes)\cup \supp(\No)$.
Furthermore,  a pair $(x,y)$ has a possibility of distinguishing (that is, the pair is interesting) only if $x,y \in C_i$ and $r_i \in \capt(x,y)$.
Focus on such interesting pairs. For such a pair, both $\sgn_{f_{\bT}}(x,y)$ and $\sgn_{g_{\bT}}(x,y)$ are equally likely to be $+1$ or $-1$.
Therefore, to distinguish, we would need two interesting pairs, $(x,y) \in C_i$ and $(x',y') \in C_j$ with $i \neq j$. Note that, when $g \sim \No$,
the signs $\sgn_{g_{\bT}}(x,y)$ and $\sgn_{g_{\bT}}(x',y')$ are independently set, whereas when $f \sim \Yes$, the signs are either the same when $r_i = r_j$, or independently set.
But if the collision event doesn't occur, we have $r_i \neq r_j$ for interesting pairs in different subcubes. Therefore, the probabilities are the same.

We now prove the lemma formally.
	Condition on a good $\bS$. Note that
	the probability of the $\Yes$ distribution depends
	solely on $\{\alpha_b\}$ and that of the $\No$ distribution
	depends solely on $\{\beta_i\}$.
	
	Consider any pair $(x,y) \in Q\times Q$ with $x \prec y$.
	We can classify it into three types: (i) $x$ and $y$ are in different
	subcubes, (ii) $x$ and $y$ are both in $C_i$, and $\coord{i}{\bS}(x,y) \neq r_i$,
	(iii) $x$ and $y$ are both in $C_i$, and $\coord{i}{\bS}(x,y) = r_i$.
	For convenience, we refer to the third type as {\em interesting pairs}.
	Let $h \in \supp(\Yes | \bS) \cup \supp(\No | \bS)$.
	For the first and second types of pairs, by \Clm{inv} and \Clm{interesting}, we have $\sgn_h(x,y) = 1$.
	For interesting pairs, by \Clm{interesting}, $\sgn_h(x,y)$
	must have the same label for all pairs in $Q_i \times Q_i$.
	Thus, any $G$ whose labels disagree with the above can never
	be $G^Q_f$ or $G^Q_g$.
	
	Fix a signed graph $G$. For any pair $(x,y) \in Q \times Q$, where $x \prec y$,
	let $w(x,y)$ be the label in $G$. Furthermore, for all interesting pairs
	in the same $Q_i$, $w(x,y)$ has the same label, denoted $w_i$.
	Let $I$ denote the set of subcubes with interesting pairs.
	At this point, all of our discussion depends purely on $\bS$
	and involves no randomness.
	
	Now we focus on $g \sim (\No|\bS)$.
	\begin{eqnarray*}
		\Pr_{g\sim (\No|\bS)}[G^Q_g = G]
		& = & \Pr\Big[\bigwedge_{i \in I} \bigwedge_{\substack{x,y \in Q_i \\ \coord{i}{\bS}(x,y) = r_i}}
		(w(x,y) = \sgn_{g_{\bT}}(x,y))\Big]  \\
		& = & \Pr\Big[\bigwedge_{i \in I} \bigwedge_{\substack{x,y \in Q_i \\ \coord{i}{\bS}(x,y) = r_i}}
		(w(x,y) = \beta_{i})\Big]  \ \ \ \textrm{(by \Clm{interesting})}\\
		& = & \Pr\Big[\bigwedge_{i \in I} (w_i = \beta_{i})\Big]
	\end{eqnarray*}
	
	\noindent Observe that each $\beta_{i}$ is chosen uniformly and independently at random from $\{-1,+1\}$, and so this 	probability is exactly $2^{-|I|}$.
	
	The analogous expressions for $f \sim (\Yes|\bS)$ yield:
	$$ \Pr_{f\sim (\Yes|\bS)}[G^Q_f = G] = \Pr \Big[\bigwedge_{i \in I} (w_i = \alpha_{r_i})\Big] $$
	Note the difference here: if multiple $r_i$'s are the same, the individual events
	are not independent over different subcubes. This is precisely what the abort and collision
	events capture. We formally argue below.
	
	Consider an interesting pair $(x,y) \in Q_i \times Q_i$. Since the abort event $\cA$ does
	not happen, $\capt(x,y) \nsubseteq R$. If $\coord{i}{\bS}(x,y) = r_i \notin \capt(x,y)$,
	then there is a coordinate of $\overline{R}$ that is more significant that $\coord{i}{\bS}(x,y)$.
	This contradicts the definition of the latter; so $r_i \in \capt(x,y) \subseteq \capt(Q_i)$.
	Equivalently, a subcube index $i \in I$ iff $r_i \in \capt(Q_i)$.
	
	Since the collision event $\cC$ does not happen, for any $j \in [m]$
	where $r_j = r_i , r_j \notin \capt(Q_j)$. Alternately, for $i,i' \in I$,
	$r_i \neq r_{i'}$. Thus, $\Pr[\bigwedge_{i \in I}(w_i = \alpha_{r_i})]
	= \prod_{i \in I}\Pr[w_i = \alpha_{r_i}] = 2^{-|I|}$.
\end{proof}

\noindent Now, we are armed to prove \Thm{tv}.

\begin{proof}[Proof of \Thm{tv}]
	Given any subset of signed graphs, $\calG$, it suffices to upper bound
	\begin{align*}
	\left|\Pr_{f\sim \Yes} [G^Q_f \in \calG] - \Pr_{f\sim \No} [G^Q_f \in \calG]\right|  
	&\leq 
    \sum_{\textrm{good } \bS} \left|\Pr[\bS]\cdot\left(\Pr_{f\sim \Yes} [G^Q_f \in \calG | \bS] - \Pr_{f\sim \No} [G^Q_f \in \calG|\bS]\right) \right| \\
    & +  \sum_{\textrm{bad } \bS} \left|\Pr[\bS]\cdot\left(\Pr_{f\sim \Yes} [G^Q_f \in \calG|\bS] - \Pr_{f\sim \No} [G^Q_f \in \calG|\bS]\right) \right|.
	\end{align*}
	The first term of the RHS is $0$ by \Lem{zero}.
    The second term is at most the probability of bad events, which is $O(\delta)$
    by \Thm{bad}.
\end{proof}

\subsection{Bounding the Probability of Bad Events: Proof of Theorem~\ref{thm:bad}}\label{sec:bad}

We prove \Thm{bad} by individually bounding $\Pr[\calA]$ and $\Pr[\calC]$.

\begin{lemma}\label{lem:A}
	If $|Q|\leq \delta d\log d$, then $\Pr[\calA] \leq d^{-1/4}$.
\end{lemma}

\begin{proof} Fix any choice of $k$ (in $\bS$).
For any pair of points $x,y \in Q$, we have $\Pr[\capt(x,y) \subseteq R] \leq (\frac{2^k}{d-5})^5$.  Since $d-5 \geq d/2$ for all $d \geq 10$ and $k \leq (\log_2d)/2$, the probability is at most $32d^{-5/2}$. For a large enough $d$, a union
bound over all pairs in $Q \times Q$, which are at most $d^2\log^2d$ in number, completes the proof.
\end{proof}

\noindent The collision event is more challenging to bound,
and is actually the heart of the lower bound.
We start by showing that, if each $Q_i$ captures
few coordinates, then the collision event has low probability. A critical point is the appearance of $d\log d$ in this bound.

\begin{lemma}\label{lem:prob}
If $\sum_i |\capt(Q_i)| \leq M$, then $\Pr[\cC] = O\left(\frac{M}{d\log d}\right)$.
\end{lemma}

\begin{proof} 
	For any $r \in [d]$, 
	define $A_r := \{j: r\in \capt(Q_j)\}$ to be the set of indices of $Q_j$'s that capture coordinate $r$.
    Let $a_r := |A_r|$. Define $n_\ell := |\{r: a_r \in (2^{\ell-1},2^{\ell}]\}|$.
	Observe that $\sum_{\ell \leq \log_2 d} n_\ell 2^\ell \leq 2\sum_{r\in [d]} a_r \leq 2M$.
	
	Fix $k$. For $r\in [d]$, we say the event $\cC_r$ occurs if
	(a) $r \in R$, and (b) there exists $i,j\in [d]$ such that $r_i = r_j = r$, and $r_i \in \capt(Q_i)$ and $r_j\in \capt(Q_j)$. By the union
    bound, $\Pr[\cC| k] \leq \sum_{r=1}^d \Pr[\cC_r | k]$.
	
	Let us now compute $\Pr[\cC_r|k]$. 
	Only sets $Q_j$'s with $j \in A_r$ are of interest, since the others do not capture $r$.
	Event $\cC_r$ occurs if at least two of these sets have $r_i = r_j = r$. Hence,
\begin{align}
\Pr[\cC_r|k] & = \Pr[r\in R]\cdot \Pr[\exists i,j\in A_r: r_i = r_j = r ~|~r\in R] \notag \\
& = \frac{2^{k}}{d}\cdot \sum_{c \geq 2} {a_r \choose c} \left(\frac{1}{2^k} \right)^c \left(1-\frac{1}{2^k} \right)^{a_r-c}. \label{eq:007}
\end{align}
A fixed $r$ is in $R$ with probability ${d-1 \choose 2^k-1}/{d\choose 2^k} = \frac{2^k}{d}$.
Given that $|R| = 2^k$, the probability that $r_i = r$ is precisely $2^{-k}$.

	If $a_r \geq \frac{2^k}{4}$, then we simply upper bound \eqref{eq:007} by  $\frac{2^k}{d}$. 
	For $a_r < \frac{2^k}{4}$, 
	we upper bound \eqref{eq:007} by
\begin{align*}
	\frac{2^k}{d} \left(1-\frac{1}{2^k}\right)^{a_r} \sum_{c\geq 2} \left( a_r \cdot \frac{1}{2^k} \cdot \left( 1-\frac{1}{2^k} \right)^{-1} \right)^c \leq \frac{2^k}{d} \sum_{c \geq 2} \left( \frac{a_r}{2^{k-1}} \right)^c \leq \frac{8a^2_r}{2^k d}.
\end{align*}
Summing over all $r$ and grouping according
to $n_\ell$, we get
\begin{align*}
\Pr[\cC|k] \leq \sum_{r=1}^d \Pr[\cC_r|k] 
\leq \sum_{r: a_r \geq 2^{k-2}} \frac{2^k}{d} + \frac{8}{d} \sum_{r: a_r < 2^{k-2}} \frac{a^2_r}{2^k} 
\leq \frac{2^k}{d} \sum_{\ell > k-2} n_\ell + \frac{8}{d}  \sum_{\ell=1}^{k-2} n_\ell 2^{2\ell - k} .
\end{align*}
Averaging over all $k$, we get
\begin{align}
	\Pr[\cC] & = \frac{2}{\log_2 d} \sum_{k=1}^{(\log_2 d)/2}\Pr[\cC | k]   \quad \leq \quad \frac{16}{d \log_2 d} \sum_{k=1}^{(\log_2 d)/2} \left( \sum_{\ell=1}^{k-2} n_\ell 2^{2\ell - k} + \sum_{\ell > k-2} n_\ell 2^k \right) \notag \\
	& = \frac{16}{d\log_2 d} \left(\sum_{\ell=1}^{(\log_2 d)/2} n_\ell \sum_{k \geq \ell + 2} 2^{2\ell - k} + \sum_{\ell=1}^{\log_2 d} n_\ell \sum_{k < \ell+2} 2^k  \right).\label{eq:008}
\end{align}

\noindent
Now, $\sum_{k\geq \ell+2} 2^{2\ell - k} \leq 2^\ell$ and $\sum_{k < \ell+2} 2^k \leq 4\cdot 2^\ell$. Substituting,
$
\Pr[\cC] \leq \frac{80}{d\log_2 d} \sum_{\ell=1}^{\log_2 d} n_\ell 2^\ell \leq \frac{160M}{d\log_2 d}
$, proving the lemma.
\end{proof}

\noindent We are now left to bound $\sum_i |\capt(Q_i)|$. This is done by the following combinatorial lemma.

\begin{lemma}\label{lem:comb}
Let $V$ be a set of vectors over an arbitrary alphabet
and any number of dimensions. For any natural number $c$
and $x,y \in V$,
let $\capt_c(x,y)$ denote the (set of) first $c$ coordinates
at which $x$ and $y$ differ. Then $|\capt_c(V)| \leq c(|V|-1)$.
\end{lemma}

\begin{proof}
We construct $c$ different edge-colored graphs $G_1, \ldots, G_c$ over the vertex set $V$. 
For every coordinate $i\in \capt_c(V)$, there must exist at least one pair of
vectors $x,y$ such that $i \in \capt_c(x,y)$. Thinking
of each $\capt_c(x,y)$ as an ordered set, find a pair
$(x,y)$ where $i$ appears ``earliest'' in $\capt_c(x,y)$.
Let the position of $i$ in this $\capt_c(x,y)$ be denoted $t$.
We add edge $(x,y)$ to $G_t$, and color it $i$.
Note that the same edge $(x,y)$ cannot be added to $G_t$
with multiple colors, and hence all $G_t$'s are simple graphs.
Furthermore, observe that each color is present only
once over all $G_t$'s.

We claim that each $G_t$ is acyclic. Suppose not. Let there be a cycle $C$ and let $(x,y)$ be the edge in $C$ with the smallest color $i$. Clearly, $x_i \neq y_i$ since $i \in \capt_c(x,y)$. There must exist another edge $(u,v)$ in $C$
such that $u_i \neq v_i$. Furthermore, the color of $(u,v)$
is $j > i$. Thus, $j$ is the $\ord{t}$ entry in $\capt_c(u,v)$.
Note that $i \in \capt_c(u,v)$ and must be the $\ord s$ entry
for some $s < t$. But this means that the edge $(u,v)$
colored $i$ should be in $G_s$, contradicting
the presence of $(x,y) \in G_t$.
\end{proof}

\noindent We wrap up the bound now.

\begin{lemma}\label{lem:C}
	If $|Q|\leq \delta d\log d$, then $\Pr[\calC] = O(\delta)$.
\end{lemma}
\begin{proof}
\Lem{comb} applied to each $Q_i$, yields
$\sum_i |\capt(Q_i)| \leq 5|Q_i| = 5|Q|$.  An application of
\Lem{prob} completes the proof.
\end{proof}

\section{Acknowledgments}
We thank Oded Goldreich for useful discussions and Meiram Murzabulatov for participation in initial discussions on this work.

\bibliography{references}

\begin{thebibliography}{10}

\bibitem{AC06}
Nir Ailon and Bernard Chazelle.
\newblock Information theory in property testing and monotonicity testing in
  higher dimension.
\newblock {\em Inf. Comput.}, 204(11):1704--1717, 2006.

\bibitem{BMPR16}
Roksana Baleshzar, Meiram Murzabulatov, Ramesh Krishnan~S. Pallavoor, and Sofya
  Raskhodnikova.
\newblock Testing unateness of real-valued functions.
\newblock {\em CoRR}, abs/1608.07652, 2016.

\bibitem{BRW05}
Tugkan Batu, Ronitt Rubinfeld, and Patrick White.
\newblock Fast approximate {PCP}s for multidimensional bin-packing problems.
\newblock {\em Inf. Comput.}, 196(1):42--56, 2005.

\bibitem{BB15}
Aleksandrs Belovs and Eric Blais.
\newblock Quantum algorithm for monotonicity testing on the hypercube.
\newblock {\em Theory of Computing}, 11:403--412, 2015.

\bibitem{BB16}
Aleksandrs Belovs and Eric Blais.
\newblock A polynomial lower bound for testing monotonicity.
\newblock In {\em Proceedings, ACM Symposium on Theory of Computing (STOC)},
  pages 1021--1032, 2016.

\bibitem{BerRY14}
Piotr Berman, Sofya Raskhodnikova, and Grigory Yaroslavtsev.
\newblock {$L_p$-testing}.
\newblock In {\em Proceedings, ACM Symposium on Theory of Computing (STOC)},
  pages 164--173, 2014.

\bibitem{BGJRW12}
Arnab Bhattacharyya, Elena Grigorescu, Kyomin Jung, Sofya Raskhodnikova, and
  David~P. Woodruff.
\newblock Transitive-closure spanners.
\newblock {\em SIAM J.\ Comput.}, 41(6):1380--1425, 2012.

\bibitem{BBM12}
Eric Blais, Joshua Brody, and Kevin Matulef.
\newblock Property testing lower bounds via communication complexity.
\newblock {\em Computational Complexity}, 21(2):311--358, 2012.

\bibitem{BlaRY14}
Eric Blais, Sofya Raskhodnikova, and Grigory Yaroslavtsev.
\newblock Lower bounds for testing properties of functions over hypergrid
  domains.
\newblock In {\em Proceedings, IEEE Conference on Computational Complexity
  (CCC)}, pages 309--320, 2014.

\bibitem{BCGM12}
Jop Bri{\"{e}}t, Sourav Chakraborty, David Garc{\'{\i}}a{-}Soriano, and Arie
  Matsliah.
\newblock Monotonicity testing and shortest-path routing on the cube.
\newblock {\em Combinatorica}, 32(1):35--53, 2012.

\bibitem{C16a}
Deeparnab Chakrabarty.
\newblock Monotonicity testing.
\newblock In {\em Encyclopedia of Algorithms}, pages 1352--1356. Springer,
  2016.

\bibitem{CDJS15}
Deeparnab Chakrabarty, Kashyap Dixit, Madhav Jha, and C.~Seshadhri.
\newblock Property testing on product distributions: Optimal testers for
  bounded derivative properties.
\newblock In {\em Proceedings, ACM-SIAM Symposium on Discrete Algorithms
  (SODA)}, pages 1809--1828, 2015.

\bibitem{CS13}
Deeparnab Chakrabarty and C.~Seshadhri.
\newblock Optimal bounds for monotonicity and {Lipschitz} testing over
  hypercubes and hypergrids.
\newblock In {\em Proceedings, ACM Symposium on Theory of Computing (STOC)},
  pages 419--428, 2013.

\bibitem{CS14}
Deeparnab Chakrabarty and C.~Seshadhri.
\newblock An optimal lower bound for monotonicity testing over hypergrids.
\newblock {\em Theory of Computing}, 10:453--464, 2014.

\bibitem{CS16}
Deeparnab Chakrabarty and C.~Seshadhri.
\newblock An $o(n)$ monotonicity tester for boolean functions over the
  hypercube.
\newblock {\em SIAM J.\ Comput.}, 45(2):461--472, 2016.

\bibitem{CS16unate}
Deeparnab Chakrabarty and C.~Seshadhri.
\newblock A $\widetilde{O}(n)$ non-adaptive tester for unateness.
\newblock {\em Electronic Colloquium on Computational Complexity {(ECCC)}},
  23:133, 2016.
\newblock Also appeared as arXiv report 1608.06980.

\bibitem{CDST15}
Xi~Chen, Anindya De, Rocco~A. Servedio, and Li{-}Yang Tan.
\newblock Boolean function monotonicity testing requires (almost)
  ${O}(n^{1/2})$ non-adaptive queries.
\newblock In {\em Proceedings, ACM Symposium on Theory of Computing (STOC)},
  pages 519--528, 2015.

\bibitem{CST14}
Xi~Chen, Rocco~A. Servedio, and Li{-}Yang Tan.
\newblock New algorithms and lower bounds for monotonicity testing.
\newblock In {\em Proceedings, IEEE Symposium on Foundations of Computer
  Science (FOCS)}, pages 286--295, 2014.

\bibitem{CWX17}
Xi~Chen, Erik Waingarten, and Jinyu Xie.
\newblock Beyond talagrand functions: New lower bounds for testing monotonicity
  and unateness.
\newblock {\em CoRR}, abs/1702.06997, 2017.
\newblock To appear in STOC 2017.

\bibitem{DRTV16}
Kashyap Dixit, Sofya Raskhodnikova, Abhradeep Thakurta, and Nithin~M. Varma.
\newblock Erasure-resilient property testing.
\newblock In {\em Proceedings, International Colloquium on Automata, Languages
  and Processing (ICALP)}, pages 91:1--91:15, 2016.

\bibitem{DGLRRS99}
Yevgeny Dodis, Oded Goldreich, Eric Lehman, Sofya Raskhodnikova, Dana Ron, and
  Alex Samorodnitsky.
\newblock Improved testing algorithms for monotonicity.
\newblock {\em Proceedings, International Workshop on Randomization and
  Approximation Techniques in Computer Science (RANDOM)}, pages 97--108, 1999.

\bibitem{EKKRV00}
Funda Erg{\"{u}}n, Sampath Kannan, Ravi Kumar, Ronitt Rubinfeld, and Mahesh
  Viswanathan.
\newblock Spot-checkers.
\newblock {\em J.\ Comput.\ System Sci.}, 60(3):717--751, 2000.

\bibitem{Fis04}
Eldar Fischer.
\newblock On the strength of comparisons in property testing.
\newblock {\em Inf. Comput.}, 189(1):107--116, 2004.

\bibitem{FLNRRS02}
Eldar Fischer, Eric Lehman, Ilan Newman, Sofya Raskhodnikova, Ronitt Rubinfeld,
  and Alex Samorodnitsky.
\newblock Monotonicity testing over general poset domains.
\newblock In {\em Proceedings, ACM Symposium on Theory of Computing (STOC)},
  pages 474--483, 2002.

\bibitem{Go-book}
Oded Goldreich.
\newblock {\em Introduction to Property Testing (working draft)}.
\newblock 2015.
\newblock URL: \url{www.wisdom.weizmann.ac.il/~oded/PDF/pt-v1.pdf}.

\bibitem{GGLRS00}
Oded Goldreich, Shafi Goldwasser, Eric Lehman, Dana Ron, and Alex
  Samorodnitsky.
\newblock Testing monotonicity.
\newblock {\em Combinatorica}, 20:301--337, 2000.

\bibitem{GGR98}
Oded Goldreich, Shafi Goldwasser, and Dana Ron.
\newblock Property testing and its connection to learning and approximation.
\newblock {\em J.\ ACM}, 45(4):653--750, 1998.

\bibitem{HK07}
Shirley Halevy and Eyal Kushilevitz.
\newblock Distribution-free property-testing.
\newblock {\em SIAM J.\ Comput.}, 37(4):1107--1138, 2007.

\bibitem{HK08}
Shirley Halevy and Eyal Kushilevitz.
\newblock Testing monotonicity over graph products.
\newblock {\em Random Struct. Algorithms}, 33(1):44--67, 2008.

\bibitem{JR13}
Madhav Jha and Sofya Raskhodnikova.
\newblock Testing and reconstruction of {Lipschitz} functions with applications
  to data privacy.
\newblock {\em SIAM J.\ Comput.}, 42(2):700--731, 2013.

\bibitem{KMS15}
Subhash Khot, Dor Minzer, and Muli Safra.
\newblock On monotonicity testing and boolean isoperimetric type theorems.
\newblock In {\em Proceedings, IEEE Symposium on Foundations of Computer
  Science (FOCS)}, pages 52--58, 2015.

\bibitem{KS16}
Subhash Khot and Igor Shinkar.
\newblock An $\widetilde{O}(n)$ queries adaptive tester for unateness.
\newblock In {\em Approximation, Randomization, and Combinatorial Optimization.
  Algorithms and Techniques, {APPROX/RANDOM} 2016}, pages 37:1--37:7, 2016.

\bibitem{LR01}
Eric Lehman and Dana Ron.
\newblock On disjoint chains of subsets.
\newblock {\em J.\ Combin.\ Theory Ser.\ A}, 94(2):399--404, 2001.

\bibitem{PRV17}
Ramesh Krishnan~S. Pallavoor, Sofya Raskhodnikova, and Nithin Varma.
\newblock Parameterized property testing of functions.
\newblock In {\em Proceedings, Innovations in Theoretical Computer Science
  (ITCS)}, 2017.

\bibitem{Ras16}
Sofya Raskhodnikova.
\newblock Testing if an array is sorted.
\newblock In {\em Encyclopedia of Algorithms}, pages 2219--2222. Springer,
  2016.

\bibitem{RS96}
Ronitt Rubinfeld and Madhu Sudan.
\newblock Robust characterizations of polynomials with applications to program
  testing.
\newblock {\em SIAM J.\ Comput.}, 25(2):252--271, 1996.

\bibitem{Yao77}
Andrew~Chi{-}Chih Yao.
\newblock Probabilistic computations: Toward a unified measure of complexity
  (extended abstract).
\newblock In {\em Proceedings, IEEE Symposium on Foundations of Computer
  Science (FOCS)}, pages 222--227, 1977.

\end{thebibliography}

\appendix
\section{Missing Details from the Main Body}
\subsection{The Lower Bound for Adaptive Testers over Hypergrids} \label{sec:adap-lb}
We show that every unateness tester for functions $f:[n]^d \mapsto \R$  requires $\Omega\left(\frac{d \log n}{\eps}-\frac{\log 1/\eps}{\eps}\right)$ queries for $\eps\in(0,1/4)$ and prove Theorem~\ref{thm:adap-lb}.

\begin{proof}[Proof of Theorem~\ref{thm:adap-lb}]
By Yao's minimax principle and the reduction to testing with comparison-based testers from~\cite{CS14} (stated for completeness in Theorem~\ref{thm:CS14}), it is sufficient to give a hard input distribution on which every deterministic comparison-based tester fails with probability more than 2/3. We use the hard distribution constructed  by Chakrabarty and Seshadhri~\cite{CS14} to prove the same lower bound for testing monotonicity. Their distribution is a mixture of two distributions, $\Yes$ and $\No,$ on positive and negative instances, respectively. Positive instances for their problem are functions that are monotone and, therefore, unate; negative instances are functions that are $\eps$-far from monotone. We show that their $\No$ distribution is supported on functions that are $\eps$-far from unate, i.e., negative instances for our problem. Then the required lower bound for unateness follows from the fact that every deterministic comparison-based tester needs the stated number of queries to distinguish $\Yes$ and $\No$ distributions with high enough probability.

We start by describing the $\Yes$ and $\No$ distribution used in \cite{CS14}. We will define them as distributions on functions over the hypercube domain. Next, we explain how to convert functions over hypercubes  to functions over hypergrids.

Without loss of generality, assume $n$ is a power of $2$ and let $\ell := \log_2 n$.
For any $z \in [n]$, let $bin(z)$ denote the binary representation of $z-1$ as an $\ell$-bit vector $(z_1, \ldots, z_\ell)$, where $z_1$ is the least significant bit.

We now describe the mapping used to convert functions on hypergrids to functions on hypercubes. Let $\phi:[n]^d \to \{0,1\}^{d\ell}$ be the mapping that takes $y\in[n]^d$ to the concatenation of $bin(y_1),\dots,bin(y_d)$. Any function $f:\{0,1\}^{d \ell} \mapsto \R$ can be easily converted into a function $\widetilde{f}:[n]^d \mapsto \R$, where $\widetilde{f}(y) := f(\phi(y))$.

Let $m:= d\ell$. For $x \in \{0,1\}^m$, let $\texttt{val}(x) = \sum\nolimits_{i=1}^m x_i 2^{i-1}$ denote the value of the binary number represented by vector $x$. For simplicity, assume $1/\eps$ is a power of $2$.
Partition the set of points $x \in \{0,1\}^m$ according to the most significant $\log(1/2\eps)$ dimensions.
That is, for $k \in \{1,2,\ldots, 1/2\eps\}$, let
$$S_k := \{x: \texttt{val}(x) \in [(k-1)\cdot \eps 2^{m+1}, k\cdot\eps 2^{m+1} - 1]\}.$$
The hypercube is partitioned into $1/2\eps$ sets $S_k$ of equal size, and each $S_k$ forms a subcube of dimension $m' = m - \log(1/\eps) + 1$.

We now describe the $\Yes$ and $\No$ distributions for functions on hypercubes.
The $\Yes$ distribution consists of a single function $f(x) = 2 \texttt{val}(x)$.
The $\No$ distribution is uniform over $m'/2\eps$ functions $g_{j,k}$, where $j \in [m']$ and $k \in [1/2\eps]$, defined as follows:
\begin{align*}
    g_{j,k}(x) =
    \begin{cases}
        2\texttt{val}(x) - 2^j - 1 &\text{ if } x_j = 1 \text{ and } x \in S_k;\\
        2\texttt{val}(x), &\text{ otherwise.}
    \end{cases}
\end{align*}
To get the $\Yes$ and $\No$ distributions for the hypergrid, we convert $f$ to $\widetilde{f}$ and each function $g_{j,k}$ to $\widetilde{g_{j,k}}$, using the transformation defined before.

Chakrabarty and Seshadhri~\cite{CS14} proved that $f$ is monotone and each function $\widetilde{g_{j,k}}$  is $\eps$-far from monotone. It remains to show that functions $\widetilde{g_{j,k}}$ are also $\eps$-far from unate.

\begin{claim}
Each function $\widetilde{g_{j,k}}$ is $\eps$-far from unate.
\end{claim}

\begin{proof}
To prove that $\widetilde{g_{j,k}}$ is $\eps$-far from unate, it suffices to show that there exists a dimension $i$, such that there are at least $\eps 2^{d\ell}$ increasing $i$-pairs and at least $\eps 2^{d\ell}$ decreasing $i$-pairs w.r.t.\ $\widetilde{g_{j,k}}$ and that all of these $i$-pairs are disjoint. Let $u,v \in [n]^d$ be two points such that $\phi(u)$ and $\phi(v)$ differ only in the $\ord{j}$ bit. Clearly, $u$ and $v$ form an $i$-pair, where $i =\lceil j/\ell \rceil$. Now, if $\phi(u),\phi(v) \in S_k$ and $u \prec v $, then $\widetilde{g_{j,k}}(v) = \widetilde{g_{j,k}}(u) - 1$. So, the $i$-pair $(u,v)$ is decreasing. The total number of such $i$-pairs is  $2^{d\ell - \log(1/2\eps) - 1} = \eps 2^{d\ell}$. If $\phi(u),\phi(v) \in S_{k'}$ where $k' \neq k$, then the $i$-pair $(u,v)$ is increasing. Clearly, there are at least $\eps 2^{d\ell}$ such $i$-pairs. All the $i$-pairs we mentioned are disjoint.
Hence, $\widetilde{g_{j,k}}$ is $\eps$-far from unate.
\end{proof}
This completes the proof of Theorem~\ref{thm:adap-lb}.
\end{proof}

\subsection{The Lower Bound for Nonadaptive Testers over Hypergrids}\label{sec:na-lb-hg}
The lower bound for nonadaptive testers over hypergrids follows from a combination of the lower bound for nonadaptive testers over hypercube and the lower bound for adaptive testers over hypergrids.

\begin{theorem}
Any nonadaptive unateness tester (even with two-sided error) for real-values functions $f:[n]^d \mapsto \R$ must make $\Omega(d(\log n + \log d))$ queries.
\end{theorem}
\begin{proof}
Fix $\eps = 1/8$. The proof consists of two parts. The lower bound for adaptive testers is also a lower bound for nonadaptive tester, and so, the bound of $\Omega(d \log n)$ holds.
Next, we extend the $\Omega(d \log d)$ lower bound for hypercubes.
Assume $n$ to be a power of $2$. Define function $\psi:[n] \mapsto \{0,1\}$ as $\psi(a):= \mathbbm{1}[a > n/2]$ for $a \in [n]$. For $x = (x_1, x_2, \ldots, x_d) \in [n]^d$, define the mapping $\Psi: [n]^d \mapsto \{0,1\}^d$ as $\Psi(x) := (\psi(x_1), \psi(x_2), \ldots, \psi(x_d))$. Any function $f:\{0,1\}^d \mapsto \R$ can be extended to $\tilde{f}:[n]^d \mapsto \R$ using the mapping $\tilde{f}(x) = f(\Psi(x))$ for all $x \in [n]^d$. The proof of Theorem~\ref{thm:non-adap-lb} goes through for hypergrids as well, and so we have an $\Omega(d \log d)$ lower bound. Combining the two lower bounds, we get a bound of $\Omega(d \cdot \max\{\log n, \log d\})$, which is asymptotically equal to $\Omega(d(\log n + \log d))$.
\end{proof}
\end{document}